\documentclass[12pt,a4paper]{article}
\newcommand{\ignore}[1]{}

\usepackage{float}
\usepackage[
  a4paper,top=3.5cm,bottom=2.5cm,left=3.5cm,right=3.5cm,marginparwidth=1.75cm]{geometry}
\usepackage{amsmath}
\usepackage{calrsfs}
\usepackage{amsthm,bm,mathtools,etoolbox,mathrsfs,amsfonts,bbm}
\usepackage{hyperref}
\hypersetup{colorlinks,linkcolor={blue},citecolor={blue},urlcolor={blue}}    
\usepackage[nameinlink,capitalise]{cleveref}
\usepackage[utf8]{inputenc} 
\usepackage[T1]{fontenc}    
\usepackage{url}            
\usepackage{booktabs}       
\usepackage{natbib}
\usepackage{placeins}
\usepackage{adjustbox}
\usepackage{booktabs}
\usepackage{algorithm}
\usepackage{algpseudocode}
\usepackage{multicol} 
\usepackage{makecell, multirow, tabularx} 
\usepackage{rotating}
\usepackage{float}
\usepackage{gensymb}
\usepackage{ulem}
\usepackage{graphicx}
\usepackage{tabularx}
\usepackage{microtype}
\usepackage{authblk}
\usepackage{rotating}
\usepackage[font=small]{caption}

\newtheorem{Theorem}{Theorem}[section]

\newtheorem{Remark}{Remark}[section]

\providecommand{\keywords}[1]
{
  \small	
  \textit{Keywords:} #1
}

\providecommand{\JELcodes}[1]
{
  \small	
  \textit{JEL Codes:} #1
}

\title{GARCHX-NoVaS: A Model-free Approach to Incorporate Exogenous Variables}

\author[1]{Kejin Wu}
\author[2]{Sayar Karmakar}
\author[3]{Rangan Gupta}

\affil[1]{Department of Mathematics, University of California San Diego}
\affil[2]{Department of Statistics, University of Florida}
\affil[3]{Department of Economics, University of Pretoria}

\date{}
\begin{document}

\maketitle

\begin{abstract}
In this work, we explore the forecasting ability of a recently proposed normalizing and variance-stabilizing (NoVaS) transformation with the possible inclusion of exogenous variables. From an applied point-of-view, extra knowledge such as fundamentals- and sentiments-based information could be beneficial to improve the prediction accuracy of market volatility if they are incorporated into the forecasting process. In the classical approach, these models including exogenous variables are typically termed GARCHX-type models. Being a Model-free prediction method, NoVaS has generally shown more accurate, stable and robust (to misspecifications) performance than that compared to classical GARCH-type methods. This motivates us to extend this framework to the GARCHX forecasting as well. We derive the NoVaS transformation needed to include exogenous covariates and then construct the corresponding prediction procedure. We show through extensive simulation studies that bolster our claim that the NoVaS method outperforms traditional ones, especially for long-term time aggregated predictions. We also provide an interesting data analysis to exhibit how our method could possibly shed light on the role of geopolitical risks in forecasting volatility in national stock market indices for three different countries in Europe.
\end{abstract}


\JELcodes{C32; C53; C63; Q54}

\section{Introduction}\label{sec:introduction}
In the long history of time series econometrics literature, accurate forecasting has always stood out as a fundamental and important problem. It has a range of applications in various industries, e.g., weather forecasting, climate forecasting, and economic forecasting. Discrete-time series data, e.g., heights of ocean tides, and temperature of a city, is the realization of a stochastic process $\{X_t,t\in \mathbb{Z}\}$. The earliest modern time series analysis could be traced back to the work of \cite{yule1927vii} where the pattern of the sunspots number was studied. Unlike the prediction of independent data, the prediction of time series gets more complicated due to the inherent data dependence. To get accurate predictions and inferences, it is crucial to model the dependent relationship within the data. Usually, very generally speaking, the time series data is assumed to be generated by some underlying mechanism as follows: 
    \begin{equation}\label{generalmodel}
        X_{t} = G(\bm{X}_{t-p},\epsilon_{t});
    \end{equation}
$G(\cdot,\cdot)$ could be any suitable function; $\epsilon_t$ is called innovation and assumed to be $i.i.d.$ with appropriate moments and independent with $X_{t-i}$, $i\geq 1$; $\bm{X}_{t-p}$ represents $\{X_{t-1},\ldots,X_{t-p}\}$ and stands for the historical information. To further simplify the forecasting problem, participators focus on some standard formats of $G(\cdot,\cdot)$, e.g., linear or non-linear. For linear models, such as linear AR, MA and ARMA models, we can apply the Box-Jenkins method of identifying, fitting, checking and predicting models systematically \citep{box1976time}. However, the prediction of non-linear models is not as trivial as the case of linear models since the innovation must be appropriately included in the prediction process, especially for the multi-step ahead predictions; see \cite{wu2023bootstrap} for more related discussions.  

In this paper, we are exclusively interested in one non-linear type of \cref{generalmodel} which is the so-called Generalized Auto-Regressive Conditional Heteroskedasticity (GARCH) model proposed by \cite{bollerslev1986generalized} and has a form below:  
\begin{equation}
    \begin{split}
        Y_t &= \sigma_tW_t\\
        \sigma_t^2&=a + a_1Y_{t-1}^2 + b_1\sigma_{t-1}^2;\label{Eq:garch}
    \end{split}
\end{equation}
where, $a \geq 0$, $a_1 > 0$, $b_1 > 0$, and $W_t\sim i.i.d.~N(0,1)$. The GARCH model is a generalization of the famous Autoregressive Conditional Heteroskedasticity (ARCH) model proposed by \cite{engle1982autoregressive}. Its ability to forecast the absolute magnitude and quantiles or entire density of squared financial log-returns (i.e., equivalent to volatility forecasting to some extent) was shown by \citet{articleengle2001} using the Dow Jones Industrial Index. Later, many studies to investigate the performance of different GARCH-type models in predicting volatility of financial series were conducted; see following references \citep{peters2001estimating,gonzalez2004forecasting,lim2013comparing, herrera2018forecasting,karmakar2020bayesian}. For ARCH/GARCH-type models, it is usual practice to identify and fit models based on quasi-maximum likelihood inference. 

Traditionally, economists primarily utilize univariate GARCH-family models to understand dynamics of econometric data such as stock/ index/ price, etc observed for a long time. However, one of the key focuses of financial econometrics is to understand how extra knowledge such as fundamentals- and sentiments-based information could be beneficial to improve the prediction accuracy of market volatility if they are incorporated into the forecasting process; see more discussion from \cite{engle2007good} and the references therein \footnote{In this regard, there is also a large literature that involves incorporating information of low-frequency variables using the GARCH-Mixed Data Sampling (MIDAS) model, as originally developed by \cite{engle2013stock}.}. In line with the classical GARCH methods, we can wrap the exogenous covariates into the prediction process to get the GARCH models augmented with additional explanatory
variables, GARCHX models in short. The estimation methodology of GARCHX models was discussed thoroughly in the work of \cite{francq2019qml}; see more details about the GARCHX model in \cref{Subsec:GARCHX}.

Rather than taking the traditional approach (i.e., specifying and fitting a model and then predicting), we consider a recently proposed model-free prediction idea, namely normalizing and variance-stabilizing transformation (NoVaS transformation). The NoVaS method was initially developed by \cite{politis2003normalizing} and then well discussed under the framework of the Model-free prediction principle in \cite{politis2015modelfreepredictionprinciple}. In short, the Model-free prediction principle hinges on the idea of applying an inverse transformation function to bridge two equivalent probability spaces. For example, if we observe a univariate time series $\{Y_1,\ldots,Y_T\}$, we can try to find a transformation to map $\{Y_1,\ldots,Y_T\}$ to an $i.i.d.$ series $\{Z_1,\ldots,Z_T\}$. Since the prediction of $i.i.d.$ data is trivial, we can then transform the prediction of $i.i.d.$ data back to the prediction of the original data; see more details about the Model-free prediction principle in \cref{Subsec:modelfreeidea}. 

Following the literature, we call this transformation-based approach a Model-free method in this paper. This necessitates us to emphasize that although the transformation is inspired from a model-assumption, for the prediction step it does not tie with any specific model structure. In the huge literature of applied econometrics, while analyzing data observed over a long time that can show signs of heteroscedasticity, the usual practice is to pick some specific GARCH-type model. Next an estimation of that model is carried out accordingly and subsequently the forecast of future volatility will be made based on the estimated model. Apparently, there is no universal rule for the choice of the specific GARCH model. In other words, a specific GARCH model can not work uniformly well across different datasets compared to other variants. On the other hand, the Model-free prediction could work well for any scenario as long as a suitable transformation function can be found. Therefore, the model-selection stage is shunned and the Model-free prediction approach is more robust against the model misspecification. Moreover, the standard GARCH methods require a relatively large sample size to be estimated well. For the NoVaS method, it tends to work stably even with short data. The existence of such transformation function in the context of predicting with exogenous variables will be analyzed theoretically in \cref{Subsec:modelfreeidea}.

Empirically speaking, NoVaS methods were mainly applied to forecast volatility in financial econometrics in the past few years. \citet{gulay2018comparison} showed that the NoVaS method could beat GARCH-type models (GARCH, EGARCH and GJR-GARCH) with generalized error distributions by comparing the pseudo-out of sample (POOS) forecasting of volatility. Here the POOS forecasting analysis means using data up to and including the current time to predict future values. Later, \citet{chen2019optimal} extended the NoVaS method to do multi-step ahead predictions. \cite{wu2021model} further substantiated the great performance of NoVaS methods on time-aggregated long-term (30-steps ahead) predictions. \cite{wang2022model} applied the Model-free idea to provide estimation and prediction inference for a general class of time series. Our present work is motivated by the \cite{wu2023model} work where the authors recommended a so-called GARCH-NoVaS (GA-NoVaS) transformation structure inspired by the development of GARCH from ARCH. This NoVaS method is significantly robust against different model misspecification. Given this, it was a natural and probably quite an important question to see if such a robust forecasting framework can be built where exogenous covariates can be included and thus improve forecasting accuracy.

In this work, we explore the new methodology of forecasting stock market volatility with additional covariates being available to be included in the volatility dynamics. As far as we know, the NoVaS model-free prediction idea has not been studied when the exogenous variables are featured even in modeling the mean or average let alone the more complicated variance or volatility dynamics of a time-series. Due to the superior performance of the GARCH-NoVaS method in volatility forecasting, for this paper, we stick to the variance part and attempt to further boost the ability of the GA-NoVaS method with the help of exogenous covariate information. Towards this, we propose a so-called GARCHX-NoVaS (abbreviated as GAX-NoVaS henceforth) method which takes the GARCHX model as the starting step to build transformation. To obtain the inference about the future situation at an overall level, we choose the time-aggregated prediction metric. This aggregated metric has been applied to evaluate future predictions of electricity price or financial data \citep{fryzlewicz2008normalized,chudy2020long, karmakar2022long}; see the formal definition in \cref{Subsec:timeaggregatedmetric}. We wanted to check if the NoVaS prediction method can incorporate exogenous variables. More importantly, we hope the GAX-NoVaS method can sustain its great performance compared to the GARCHX method. 

In addition to comparing our model-free method and classical GARCHX model with several simulated datasets, we also provide an interesting real data analysis. Our goal is to exhibit how our method could possibly shed light on the role of geopolitical risks, which are currently engulfing the global economy with multiple wars taking place, in forecasting the volatility of three stock markets of Europe: Germany and its two neighbors (Austria and Switzerland), based on a daily index of uncertainty associated with the Russia-Ukraine war as perceived by German Twitter activity. We also use newspaper-based metrics of global geopolitical risks due to acts and threats, as developed by \cite{caldara2022}, to check for the robustness of our result covering a longer data sample. Hence, we add from a methodological perspective to the existing literature on forecasting international stock returns volatility using the information contained in geopolitical events and threats that basically rely on GARCH-type models (see, for example, \cite{Salisu2022,zhang2023} for details discussion of this literature). In this regard, note that, \cite{caldara2022} pointed out that entrepreneurs, market participants, and central bank officials view geopolitical risks as key determinants of investment decisions and stock market dynamics, with such risks, along with economic and policy uncertainties, forming an ``uncertainty trinity” that would adversely impact the economy and the financial sector, as has been traditionally reported in the large existing literature on the impact of terror attacks and threats \cite{balcilar2018,bouras2018,bourforthcoming}. In our empirical and simulation exercises, we measure the performance of GARCHX and GAX-NoVaS methods by the standard mean square prediction error. Moreover, we apply the forecast comparison tests to compare the two methods in a statistical way.

Our main contributions are summarized as follows: 
\begin{itemize}
    \item We propose a new methodology-- namely GAX-NoVaS--to do the volatility forecasting with exogenous variables. This model-free method depends on a transformation function to connect two equivalent probability spaces instead of relying on any model assumption. The idea behind the GAX-NoVaS method hinges on the Model-free prediction principle. 
    
    \item We show such a transformation function exists under some mild conditions. This serves as the theoretical foundation of our method.

    \item We apply our new method and standard GARCHX model to investigate the role of geopolitical risks in forecasting volatility. It turns out that our new method can be significantly more accurate, especially for long-horizon time aggregated predictions.
\end{itemize}

We organize the remainder of this article as follows. In \cref{Sec:GARCHXandMF}, we review the classical forecasting model, namely GARCHX, which is used as the starting point to propose the GAX-NoVaS method. Also, we present more details of the Model-free prediction principle and prove the existence of a transformation function with some exogenous variables existing. In \cref{Sec:NoVaSGARCHX}, we delineate the details of proposed GAX-NoVaS method. Then, some simulation studies and model evaluation criteria are collated in \cref{Sec:Simulations}. Next, we contrast our methods to existing classical ones on three empirical datasets in \cref{Sec:empiricalstudies}.  Finally, in \cref{Sec:Conclusion}, we conclude by discussing the implications of our findings and some future directions.

\section{GARCHX estimation and Model-free prediction principle}\label{Sec:GARCHXandMF}
Before introducing our GAX-NoVaS method, we first explain the GARCHX model since the transformation of GAX-NoVaS is based on the GARCHX model. In addition, we give more details on the Model-free prediction principle and we prove the existence of a transformation function to achieve the Model-free prediction goal.

\subsection{GARCHX model}\label{Subsec:GARCHX}
In a seminal work, ARCH was proposed by \cite{engle1982autoregressive} to model volatility or $\sigma_t^2$ for a time-series in a dynamic way. Following this, many different variants were developed in the econometrics literature. The GARCH model, especially the GARCH(1,1), stands as possibly the most popular one. The classic GARCH(1,1) model as defined by \cite{bollerslev1986generalized} can be described below:
\begin{equation}\label{Eq:GARCHmodel}
    \begin{split}
        Y_t &= \sigma_tW_t,\\
        \sigma_t^2&=a + a_1Y_{t-1}^2 + b_1\sigma_{t-1}^2;
    \end{split}
\end{equation}
where $a \geq 0$, $a_1 > 0$, $b_1 > 0$, and $W_t\sim i.i.d.~N(0,1)$. More generally, we can express the GARCH-type models as 
\begin{equation*}
    Y_t = \sigma_t \eta_t;
\end{equation*}
where a particular distribution of $\eta_t$ is not necessary and we can only assume that $\mathbb{E}(\eta_t^2|\mathcal{F}_{t-1}) = 1$. Usually, $\mathcal{F}_{t-1}$ is taken as the sigma-field generated by previous information $\{Y_{t-1},\ldots\}$. When additional information is available, people would like to utilize this extra knowledge to improve prediction accuracy. Subsequently, the so-called GARCHX model enters the public eye; see \cite{francq2019qml} for discussions on the quasi-maximum likelihood estimation inference of GARCHX models. To simplify the analysis, we still assume the normality of $W_t$. After taking a vector of exogenous covariates $\bm{X} = (X_{1},\ldots,X_{m})$ into account, we can wrap the exogenous covariates into the prediction process by turning the GARCH(1,1) model into the following GARCHX(1,1,1) model:
\begin{equation}\label{Eq:GARCHXmodel}
    \begin{split}
        Y_t &= \sigma_tW_t,\\
        \sigma_t^2&=a + a_1Y_{t-1}^2 + b_1\sigma_{t-1}^2 + \bm{c}^T\bm{X}_{t-1};
    \end{split}
\end{equation}
where $\bm{X}_{t-1}$ represents $(X_{1,t-1},\ldots,X_{m,t-1})$ and $\bm{c}$ are the coefficients of these exogenous variables to be estimated. To perform a moving-window out-of-sample prediction experiment, we first need to estimate the GARCH(1,1) and GARCHX(1,1,1) models\footnote{For estimation of the GARCH and GARCHX models, we use the \textit{fGarch} \citep{wuertz2013package} and \textit{garchx} packages \citep{sucarrat2020garchx} in the \textit{R} language and environment \citep{Rlanguage}.}, and then we compute predictions iteratively; see \cref{Sec:Simulations} for details. In this process, we assume that we know the true exogenous variables, which is feasible because we generate out-of-sample predictions. For practical applications, if needed, the future exogenous information can be estimated separately. 

\subsection{Model-free prediction principle}\label{Subsec:modelfreeidea}
The model-free prediction principle was initially well developed by \cite{politis2015modelfreepredictionprinciple}. Later, \cite{chen2019optimal} applied this idea to multi-step ahead predictions of financial returns in the context of an ARCH-model structure. In short, the main idea behind the model-free prediction is to apply an invertible transformation function, $H_T$, that can map a non-$i.i.d.$ vector, $\{Y_t~;t = 1,\ldots,T\}$, to a vector, $\{\epsilon_t;~t=1,\ldots,T\}$, with $i.i.d.$ components (chosen as standard normal in this work). Due to the invertibility of the function, $H_T$, it is possible to construct a one-to-one relationship between a future value, $Y_{T+1}$, and $\epsilon_{T+1}$, i.e.,
\begin{equation}
    Y_{T+1}=f_{T+1}(\bm{Y}_{T}, \bm{X}_{T+1},\epsilon_{T+1});
\end{equation}
where $\bm{Y}_{T}$ denotes all historical data $\{Y_t;~t =1,\ldots,T\}$; $\bm{X}_{T+1}$ is the collection of all predictors, and it also contains the value of a future predictor $X_{T+1}$; the form of $f_{T+1}(\cdot)$ depends on $H^{-1}_{T}$. This relationship implies that we can also transform the prediction of $\epsilon_{T+1}$ to the prediction of $Y_{T+1}$. Assume we have $\hat{\epsilon}_{T+1}$ to the be the predictor of $\epsilon_{T+1}$, we can express the predictor of $Y_{T+1}$ as
\begin{equation}
     \widehat{Y}_{T+1}=f_{T+1}(\bm{Y}_{T}, \bm{X}_{T+1},\hat{\epsilon}_{T+1}).
\end{equation}
Because the prediction of $i.i.d.$ data is standard, the $L_1$ (Mean Absolute Deviation) , $L_2$ (Mean Squared Error), or another optimal quantile predictor of $\epsilon_{T+1}$ can easily be found. We, thus, can easily obtain the corresponding optimal predictor of $Y_{T+1}$. 

For multi-step ($h$-step) ahead prediction, we simply repeat this prediction process, i.e., we express $Y_{T+h}$ through a function w.r.t. $\bm{Y}_{T}$, $\bm{X}_{T+1}$ and $\{\epsilon_{T+1},\ldots,\epsilon_{T+h}\}$: 
\begin{equation}\label{Eq:krelationship}
    Y_{T+h}=f_{T+h}(\bm{Y}_{T}, \bm{X}_{T+1},\epsilon_{T+1},\ldots,\epsilon_{T+h}). 
\end{equation}
In order to compute the prediction of $Y_{T+h}$, we take a distribution-match approach to approximate the distribution of $Y_{T+h}$. Ideally, when we know the exact distribution of the $i.i.d.$ $\epsilon$, we can use a Monte Carlo simulation to approximate the distribution of $Y_{T+h}$ based on \cref{Eq:krelationship}. Practically speaking, when we just have the empirical transformation results, i.e., the observed sample $\{\epsilon_t\}_{t=1}^{T}$, bootstrap is an appropriate approach. Moreover, we can even predict $g(Y_{T+h})$, where $g(\cdot)$ is a general continuous function. For example, we can compute the $L_1$ and $L_2$ optimal predictors of $g(Y_{T+h})$ as below:
\begin{equation}
\begin{split}
    g(Y_{T+h})_{L_2} &= \frac{1}{M}\sum_{m=1}^Mg(f_{T+h}(\bm{Y}_T,\bm{X}_{T+1},\hat{\epsilon}_{T+1,m},\ldots,\hat{\epsilon}_{T+h,m})),\\
    g(Y_{T+h})_{L_1} &= \text{Median~of~}\{g(f_{T+h}(\bm{Y}_T,\bm{X}_{T+1},\hat{\epsilon}_{T+1,m},\ldots,\hat{\epsilon}_{T+h,m}));m = 1,\ldots,M\};
    \end{split}
\end{equation}
where $g(Y_{T+h})_{L_2}$ and $g(Y_{T+h})_{L_1}$ represent the optimal $L_2$ and $L_1$ predictor of $g(Y_{T+1})$, the $\{\hat{\epsilon}_{T+1,m}\}_{m=1}^{M}$ are generated by bootstrap or Monte Carlo simulation, and $M$ is some large number (2000 in our empirical analysis). For further discussion, see \citet{politis2015modelfreepredictionprinciple}. 

To the best of our knowledge, the model-free prediction idea has not been studied when the model features exogenous variables. However, it may be beneficial to take into account in the prediction process the additional information embedded in such exogenous variables. To show the NoVaS approach is still applicable, we need a transformation function that maps the targeted variables and exogenous predictors together into some simple $i.i.d. $ random variables. Under some mild conditions, we show the existence of such a transformation function based on the probability integral transform. We assume:
\begin{itemize}
    \item A1 The joint density of $\{Y_1,\cdots, Y_T \}$ exists for any $T\geq 1$.
    \item A2 For exogenous random vector $\bm{X}: = \{X_1,\cdots, X_{m}\}$, the joint density $\{Y_1,\cdots, Y_T, X_1, \cdots, X_m \}$ exists for any $m\geq 1$. 
\end{itemize}
Then, the feasibility of NoVaS transformation with exogenous variables existing is guaranteed by \cref{Theorem:existencetransformation} shown below:

\begin{Theorem}\label{Theorem:existencetransformation}
Under A1 and A2, there exists a function $\bm{g}$ such that $\bm{Z} = \bm{g}((\bm{Y},\bm{X}))$ and the corresponding inverse function $\bm{h}$ such that $(\widetilde{\bm{Y}},\widetilde{\bm{X}}) = \bm{h}(\bm{Z})$; $\bm{Z} \sim N(0,\bm{I}_{T+m})$; $\bm{Y} = (Y_1,\cdots,Y_T)$ and $\bm{X} = (X_1,\cdots,X_m)$ are any two random vectors; $(\widetilde{\bm{Y}},\widetilde{\bm{X}})$ have the same joint distribution of $(\bm{Y},\bm{X})$. 
\end{Theorem}

\begin{proof}
The proof of \cref{Theorem:existencetransformation} is based on the probability integral transform; see \cite{angus1994probability} for a review. Without loss of generality, we start from $Y_1$ to determine the transformation function $\bm{g}$. Let $U_1 := \Tilde{g}_1(Y) = F(Y_1)$; $F(Y_1)$ is the distribution of $Y_1$. According to the probability integral transform, we know $U_1$ has a uniform distribution on $[0,1]$. Then, we make 
\begin{equation}\label{Eq:z2}
    U_2 := \Tilde{g}_2(Y_1,Y_2) = F(Y_2 | Y_1 ).
\end{equation}
$F(Y_2 | Y_1 )$ is the conditional distribution of $Y_2$. \cref{Eq:z2} implies that $Z_2$ is $\text{Uniform}(0,1)$ conditional on $Y_1$. Thus, the unconditional (marginal) distribution of $U_2$ is still $\text{Uniform}(0,1)$, and $U_2$ and $Y_1$ are independent so that $U_2$ is also independent with $U_1$. This can be seen from the equation below:
\begin{equation}
    p_{U_2, Y_1}(u_2,y_1) =  p_{U_2|Y_1}(u_2|y_1) p_{Y_1}(y_1).
\end{equation}
Integrating both sides w.r.t. $y_1$, we can find $p_{U_2}(u_2) = 1$ on the region $[0,1]$, since $Z_2$ is $\text{Uniform}(0,1)$ conditional on $Y_1 = y_1$ for any $y_1$. We can repeat this process as a Gram–Schmidt-like recursion, i.e., we let $U_3 := \Tilde{g}_3(Y_1,Y_2,Y_3) = F(Y_2 | Y_1,Y_2 )$ and so on. In total, we need $\{\Tilde{g}_1,\cdots,\Tilde{g}_{T+m}\}$ and they are functions of $\bm{Y}$. Thus, there exists a function $\Tilde{\bm{g}}$ which maps $(\bm{Y},\bm{X})$ to $\bm{U}$ which has $i.i.d.$ uniform components $\{U_1,\cdots, U_{m+T}\}$, i.e., $\bm{U} = \Tilde{\bm{g}}( (\bm{Y},\bm{X}))$. Then, $\bm{Z} = \Phi^{-1}\circ\Tilde{\bm{g}}((\bm{Y},\bm{X}))$ has multivariate normal distribution $N(0,\bm{I}_{T+m})$; $\Phi^{-1}$ is the quantile function of $N(0,\bm{I}_{T+m})$. Finally, we can take $\bm{g} =  \Phi^{-1}\circ\Tilde{\bm{g}}$.

On the other hand, if $F^{-1}_{Y_1}:=\inf \{x: F_{Y_1}(x) \geq y\}, 0\leq y\leq 1$, then $F_{Y_1}^{-1}(U_1)$ has the distribution as the same as $Y_1$. Similarly, we can get the conditional distribution of $Y_2$ on $Y_1$ by taking $F_{Y_2|Y_1}^{-1}(U_2)$. By repeating this process, we can recover the joint distribution of $(\bm{Y},\bm{X})$ by chain rule. In other words, there exists a $\bm{h}$ such that $(\bm{Y},\bm{X}) = \bm{h}(\bm{Z})$.
\end{proof}

The direct implication of \cref{Theorem:existencetransformation} is that the Model-free prediction principle is feasible even when exogenous variables are included in the dependence dynamics. Moreover, our theorem is more general than the result from \cite{wang2022model} where the time series must satisfy some strict conditions. In short, we build a transformation function based on the GARCHX model structure to estimate the oracle functions $\bm{g}$ and $\bm{h}$, so we call our method GAX-NoVaS. We should mention again that the ``model-free'' in this context means we do not rely on an assumed underlying model to make predictions. Although a transformation function needs to be estimated, it is just a ``bridge'' that connects original and transformed distributions according to the distribution match idea.


\section{GAX-NoVaS model-free prediction method}\label{Sec:NoVaSGARCHX}
We first present the state-of-the-art GARCH-NoVaS method which is based on a so-called NoVaS transformation. Then, we extend the GARCH-NoVaS method to a GAX-NoVaS method, which features the exogenous variables.

\subsection{NoVaS transformation}\label{subSec:NoVaSidea}
For the sake of completeness, we first give a brief introduction to the NoVaS transformation (model) which is a direct application of the Model-free prediction idea explained in \cref{Subsec:modelfreeidea}. Initially, the NoVaS transformation is developed from the ARCH model:
\begin{equation}
    Y_t = W_t\sqrt{a+\sum_{i=1}^pa_iY_{t-i}^2}. \label{3.2e1}
\end{equation}
here, these parameters satisfy $a\geq 0$, $a_i\geq 0$, for all $i = 1,\ldots,p$; $W_t\sim i.i.d.~N(0,1)$. In other words, the structure of the ARCH model gives us a ready-made $H^{-1}_T$. We can express $W_t$ in \cref{3.2e1} using other terms to get a potential $H_{T}$ :
\begin{equation}
    W_t = \frac{Y_t}{\sqrt{a+\sum_{i=1}^pa_iY_{t-i}^2}} ~;~\text{for}~ t=p+1,\ldots,T. \label{3.2e2}
\end{equation}
\citet{politis2003normalizing} further modified \cref{3.2e2} as follows:
\begin{equation}
      W_{t}=\frac{Y_t}{\sqrt{\alpha s_{t-1}^2+\beta Y_t^2+\sum_{i=1}^pa_iY_{t-i}^2}}~;~\text{for}~ t=p+1,\ldots,T. \label{3.2e3}
\end{equation}
here, $\{Y_t;~t=1,\ldots,T\}$ is the sample data; $\{W_{t};~t=p+1,\ldots,T\}$ is the transformed vector; $\alpha$ is a fixed scale invariant constant; $s_{t-1}^2$ is an estimator of the variance of $\{Y_i;~i = 1,\ldots,t-1\}$ and can be calculated by $(t-1)^{-1}\sum_{i=1}^{t-1}(Y_i-\overline{Y})^2$, where $\overline{Y}$ is the sample mean of $\{Y_i;~i = 1,\ldots,t-1\}$. For making \cref{3.2e3} be a qualified function $H_T$, i.e., making $\{W_t\}_{t=p+1}^{T}$ really obey $i.i.d.$ standard normal distribution, we need to impose some restrictions on $\alpha$ and $\beta, a_1,\ldots,a_p$. We first stabilize the variance by requiring:
\begin{equation}
    \alpha\geq0, \beta\geq0, a_i\geq0~;~\text{for all}~i\geq1, \alpha + \beta + \sum_{i=1}^pa_i=1. \label{3.2e4}
\end{equation} 
In application, $\{W_t\}_{t=p+1}^{T}$ transformed from financial log-returns by NoVaS transformation are usually uncorrelated. Therefore, if we make $\{W_t\}_{t=p+1}^{T}$ close to a Gaussian series i.e., normalizing $\{W_t\}_{t=p+1}^{T}$, we can get the desired $i.i.d.$ property. This is why this transformation is called NoVaS.

There are many criteria to measure the normality of a series. Under the observation that the distribution of financial log-returns is usually symmetric, we choose the kurtosis to be a simple distance to measure the departure of a non-skewed dataset from that of the standard normal distribution \citep{politis2015modelfreepredictionprinciple}. Besides, matching marginal distribution seems sufficient to normalize the joint distribution of $\{W_t\}_{t=p+1}^{T}$ for practical purposes based on empirical results. If we denote the marginal distribution of $\{W_t\}_{t=p+1}^{T}$ and the corresponding kurtosis by $\widehat{F}_w$ and $\text{KURT}(W_t)$, respectively, we then attempt to minimize $|\text{KURT}(W_t)-3|$ to obtain the optimal combination of $\alpha,\beta,a_1,\ldots,a_p$ such that $\widehat{F}_w$ is as close to standard normal distribution as possible. Subsequently, the NoVaS transformation can be determined.

The remaining difficulty is how to finish this optimization step to get optimal coefficients $\alpha,\beta,a_1,\ldots,a_p$, especially when $p$ is large. To simplify this problem, \cite{politis2015modelfreepredictionprinciple} defined an exponentially decayed form of $\{a_i\}_{i=1}^p$:
\begin{equation}
    \alpha \neq 0, \beta = c', a_i = c'e^{-ci}~;~\text{for all}~1\leq i\leq p, c' = \frac{1-\alpha}{\sum_{j=0}^pe^{-cj}}.\label{3.2e6}
\end{equation}
The NoVaS transformation based on coefficients defined in \cref{3.2e6} is called Generalized Exponential NoVaS (GE-NoVaS). In other words, we can represent the $p+2$ number of coefficients by two parameters $c$ and $\alpha$, which relieve the optimization burden, but with a sacrifice that the coefficients are fixed in a decayed form. To achieve a balance between the relief of the optimization dilemma and the freedom of coefficients, inspired by the development of GARCH from ARCH, \cite{wu2023model} built a NoVaS transformation according to the GARCH model, namely GARCH-NoVaS which was shown to be more stable and accurate. Later, we specify the GAX-NoVaS transformation in detail.

\subsection{GAX-NoVaS transformation method}
Starting from \cref{Eq:GARCHXmodel}, we take similar steps of building GA-NoVaS to find the transformation function of the GAX-NoVaS method. In order to simplify the notation, we consider the case of only one exogenous covariate $X_t$. The case of multiple exogenous covariates can be analyzed in an analogous way. First, we notice that we can rewrite the \cref{Eq:GARCHXmodel} as
\begin{equation}\label{Eq:start}
    W_t = \frac{Y_t}{\sqrt{a_0 + a_1 Y^2_{t-1} + b_1 \sigma_{t-1}^2 + c_1 X_{t-1} }}.
\end{equation}
We also have
\begin{equation}
\begin{split}
    & \sigma_{t-1}^2 = a_0 + a_1 Y^2_{t-2} + b_1 \sigma_{t-2}^2 + c_1 X_{t-2},\\
    &\sigma_{t-2}^2 = a_0 + a_1 Y^2_{t-3} + b_1 \sigma_{t-3}^2 + c_1 X_{t-3},\\
    \vdots  
\end{split}
\end{equation}
so that we can substitute these terms into \cref{Eq:start}. We then get
\begin{equation}
     W_t = \frac{Y_t}{\sqrt{\frac{a_0}{1-b_1} + \sum_{i=1}^{p}a_1b_1^{i-1} Y^2_{t-i} + \sum_{i=1}^{p}c_1b_1^{i-1} X_{t-i} }},
\end{equation}
where $p$ is a large constant that is used to truncate the infinite summation (we take $p=q$ in this work), because $a_1$ and $b_1$ need to be less than one to guarantee the stationary property of the GARCHX series. In line with the NoVaS transformation, we finally write the transformation function as follows:
\begin{equation}\label{Eq:transfom}
    W_t = \frac{Y_t}{\sqrt{\alpha s_{t-1,Y} + \beta s_{t-1,X}+ \sum_{i=1}^{p}a_1b_1^{i-1} Y^2_{t-i} + \sum_{i=1}^{p}c_1b_1^{i-1} Y_{t-i}}},
\end{equation}
where $s^2_{t-1,Y}$ and $s^2_{t-1,X}$ are the sample variance of $\{Y_1,\ldots,Y_{t-1}\}$ and $\{X_1,\ldots,X_{t-1}\}$, respectively. Thus, we can use \cref{Eq:transfom} as the transformation function for the GARCHX model, where $\{W_t\}$ is the transformed series. We want to make $\{W_t\}$ $i.i.d.$ normal, so that the one-step (conditional) prediction $\widehat{Y}_{T+1}$ can be expressed as
\begin{equation}\label{Eq:pre}
    \widehat{Y}_{T+1} = \widehat{W}_{T+1}\sqrt{\alpha s^2_{T-1,Y} + \beta s^2_{T-1,X}+ \sum_{i=1}^{p}a_1b_1^{i-1} Y^2_{T-i} + \sum_{i=1}^{p}c_1b_1^{i-1} X_{T-i}},
\end{equation}
where $\widehat{W}_{T+1}$ is the optimal point prediction based on $i.i.d.$ $\{W_1,\ldots, W_{T}\}$. Multi-step-ahead predictions can be computed as explained in \cref{Subsec:modelfreeidea}.

The final question left now is how to find a transformation function that indeed makes $\{W_t\}$ $i.i.d.$ normal. While we have made some brief remarks on this question in \cref{subSec:NoVaSidea}, we next provide a full explanation with a focus on the GAX-NoVaS method. Our goal is to determine the coefficients, $\alpha,\beta,a_1,b_1,c_1$, of \cref{Eq:transfom} to obtain the desired transformation. The most important step is to minimize $|\text{KURT}(W_t) - 3|$, where $3$ is the kurtosis of normal distribution. For this optimization, we use the numerical technique to find the optimal coefficients.\footnote{We use the \textit{nloptr} package \citep{ypma2014nloptr} for \textit{R}.} In operation, one may obtain some extremely large values from the transformed series $\{W_t\}$, and such outliers may spoil the normality of the transformed series and may also influence prediction performance. Thus, before moving on to the prediction step, we truncate the transformed series by the 0.99 and 0.01 quantile values of a normal distribution with mean and standard deviation given by the sample mean and sample standard deviation of $\{W_t\}$.

\begin{Remark}
It is not difficult to perceive that the transformed series $\{W_t\}$ may be correlated, such that the decorrelation step is beneficial and necessary. One way to carry out this step is by fitting an AR(p) model on the $\{W_t\}$ series. Then, we record the residuals of the AR fit as $\{\hat{\epsilon}_t\}$. Also, we can approximate the one-step-ahead value $\widehat{W}_{T+1}$ with a fitted AR model. Then, we can create the new series as $\{\epsilon_{t} + \widehat{W}_{T+1} \}$. We use the empirical distribution of this new series to approximate the distribution of $W_{T+1}$. \cref{Eq:pre} can be used with the optimal prediction $\widehat{W}_{T+1}$ derived from this empirical distribution. It is still an open question, however, how to extend this decorrelation step to multi-step-ahead predictions. 
\end{Remark}

\section{Simulations}\label{Sec:Simulations}
In this section, we deploy several simulations to check the performance of GARCH, GARCHX and GAX-NoVaS methods. Before presenting the data-generating model used to do simulations, we explain the procedure of the moving-window time-aggregated predictions and give the model evaluation metrics to measure the performance of different methods. 

\subsection{Moving-window time-aggregated prediction}\label{Subsec:timeaggregatedmetric}
If we have sample $\{Y_1,\ldots,Y_{N}\}$ at hand, in order to fully exhaust the dataset, we can focus on moving-window out-of-sample predictions, i.e., we use $\{Y_1,\cdots,Y_{T}\}$ to predict \{$Y_{T+1}^2,\cdots,Y_{T+h}^2\}$, then we use $\{Y_2,\cdots,Y_{T+1}\}$ to predict $\{Y_{T+2}^2,\cdots,Y_{T+h+1}^2\}$, and so on until we reach the end of the sample (that is, until we use  $\{Y_{N-T+h+1},\cdots,Y_{N-h}\}$ to predict $\{Y^2_{N-h+1},\ldots, Y^2_{N}\}$). Here, $T$ denotes the moving window size; we fix its size as 250; $h$ is the prediction horizon, i.e., 1, 5, 20 in our setting. Sometimes, we may not have enough data available to perform predictions. Thus, the window size $T = 250$ is designed to see if this method is stable and can still return accurate predictions even with short data. A 250-size moving window is in line with around one year of daily financial data. In this perspective, it is important to keep in mind that a 250-size moving window is practically meaningful since the time series may not be stationary on a wider span. 

In addition, we are not only interested in the one-step-ahead prediction $h=1$ but also multi-step-ahead prediction $h>1$. From a practical aspect of forecasting volatility, as mentioned above in our introduction Section \ref{sec:introduction}, the long-term prediction ($h$ takes a large value) is important and can guide future strategic decisions. Before this evaluation, we start by writing time-aggregated predictions as follows:

\begin{equation}\label{Eq:timeaggregated}
    \overline{\widehat{Y}}_{T,h}^2 = \sum_{k=1}^h \widehat{Y}_{T+k}^2/h,
\end{equation}
where $\overline{\widehat{Y}}_{T,h}^2$ is the $h$-step ahead time-aggregated volatility prediction starting from $Y_T$. For example, if the total number of data $N = 1000$ and we consider the 6-step-ahead moving window time-aggregated predictions with $T = 500,$ we need to find predictions $\overline{\widehat{Y}}_{T,6}^2$ for $T = 500,\ldots,994$.
  
We hope the time-aggregated volatility prediction is close to the true aggregated value calculated from the realized average squared log-returns $\overline{Y}_{l,h}^2 = \sum_{k=1}^h(Y_{T+k}^2/h)$. To evaluate the accuracy, we can consider the specific mean of squared prediction errors (MSPE) shown below, with this statistic aiming to compare the prediction performance in an absolute way:
\begin{equation}\label{Eq:SSPE}
    P = \frac{1}{N-h-T+1}\sum_{l= T}^{N-h}(\overline{\widehat{Y}}_{l,h}^2-\overline{Y}_{l,h}^2)^2,
\end{equation}
where $\overline{\widehat{Y}}_{l,h}^2$ and $\overline{Y}_{l,h}^2$ denote the predicted and true time-aggregated values for each moving-window forecasting, respectively. 

\subsection{Simulation setting}
In this part, we present three data-generating models to simulate data and then evaluate the performance of various methods considered in this paper. Due to the true data-generating process being known to us, we can simulate any size of the sample and compare predictions from various methods with oracle values. Besides, we remove the variance term $\beta s^2_{T-1,X}$ of \cref{Eq:transfom} when we do the transformations since the prediction performance hardly changes with or without $\beta s^2_{T-1,X}$. Three true underlying models are presented below:

\begin{itemize}
\item[] \textbf{Model 1:} Standard GARCH(1,1) with Student-$t$ errors\\
$Y_t = \sigma_t\epsilon_t,$ $~\sigma_t^2 = 0.00001 + 0.73\sigma_{t-1}^2+0.1Y_{t-1}^2 + c|X_{t-1}|,$\\ $X_{t-1} \sim i.i.d.~N(0,1)$;$~\{\epsilon_t\}\sim i.i.d.~t$ $\text{distribution with four degrees of freedom}$\\
$c = 1$.

\item[] \textbf{Model 2:} Standard GARCH(1,1) with Student-$t$ errors\\
$Y_t = \sigma_t\epsilon_t,$ $~\sigma_t^2 = 0.00001 + 0.8895\sigma_{t-1}^2+0.1Y_{t-1}^2 + c|X_{t-1}|,$\\ $X_{t-1} \sim i.i.d.~N(0,1)$;$~\{\epsilon_t\}\sim i.i.d.~t$ $\text{distribution with four degrees of freedom}$\\
$c = 1$.

\item[] \textbf{Model 3:} Time-varying GARCHX(1,1) with standard normal exogenous variables:\\
    $Y_t = \sigma_t\epsilon_t,~\sigma_t^2 =  b_{t}\sigma_{t-1}^2+a_{t}Y_{t-1}^2 + c|X_{t-1}|,$\\
$X_{t-1} \sim i.i.d.~N(0,1)$;~$\{\epsilon_t\} \sim i.i.d.$ $t$ distribution with five degrees freedom;\\
$c = 1; g_t = t/n$; $a_{t} = 0.1 - 0.05g_t$; $b_{t} = 0.7 + 0.2g_t,~n$~is the total length of the time series.
\end{itemize}
To check the robustness of methods on model misspecification, we intend to take the innovation distribution of simulation models as the $t$-distribution. Besides this purpose, we argue that the $t$ distribution as the innovation to mimic the real-world cases is more appropriate since real data usually show the heavy tail phenomenon. Models 1 and 2 are from a standard GARCH where in Model 2 we intended to explore a scenario that $\alpha_1 + \beta_1$ is very close to 1 and thus mimics what would happen for the iGARCH situation. In addition, we make the coefficients of the GARCHX model change linearly in Mode-3 so that we can observe the ability of different methods to handle the data generated from a time-varying model which is more coherent to the real-world situation. To sync with the empirical studies later, we simulate a time series with a length $T = 4694$. We take the moving-window size $T = 250$. MSPE of GARCH, GARCHX and GAX-NoVaS methods with three simulation settings are presented below:
\begin{table}[htbp]
  \caption{MSPE ratios of different methods on three simulated datasets.}
  \resizebox{\textwidth}{!}{
    \begin{tabular}{lccccccccc}
    \toprule
          & \multicolumn{3}{c}{\textbf{Model-1}} & \multicolumn{3}{c}{\textbf{Model-2}} & \multicolumn{3}{c}{\textbf{Model-3}} \\
    \hline
    \textbf{Prediction steps} & 1     & 5     & 20    & 1     & 5     & 20 & 1     & 5     & 20 \\
  \textbf{GARCH} & 1.000 & 1.000 & 1.000 & 1.000 & 1.000 & 1.000 & 1.000 & 1.000 & 1.000 \\ 
  \textbf{GARCHX} & 0.985 & 1.007 & 0.854  & 1.014 & 0.957 & 0.701 & 0.921 & 0.956 & 2.680\\ 
  \textbf{GAX-NoVaS} & 0.842 & 0.566 & 0.222 & 0.798 & 0.818 & 0.085 & 0.870 & 1.276 & 0.006\\ 
    \bottomrule
    \end{tabular}}%
    
    {\raggedright
    \noindent{\footnotesize{Note: To simplify the presentation, we compute the ratio of MSPE of different methods, i.e., we divide all MSPE for each prediction horizon and models by the MSPE of the GARCH (benchmark) method. In addition, $0.000$ value in the above table is a rounding number.}} }
  \label{Table:simulationresults}%
\end{table}%

From \cref{Table:simulationresults}, it is clear that the GAX-NoVaS is much better than GARCH and even the classical GARCHX models according to the MSPE criterion. Interestingly, the GARCHX model is even much worse than the GARCH model for some specific cases, e.g., the 20-step-ahead prediction of Model-3. By taking a deeper analysis, we find the terrible performance GARCHX method is due to some extremely large predictions. On the other hand, the prediction returned by GAX-NoVaS is more stable. The superiority is further shown in \cref{Sec:empiricalstudies} with three real datasets and various exogenous predictors. 

Since the focus of this paper is exploring a new approach to incorporate exogenous variables, we take the DM test to evaluate the performance of GARCHX and GAX-NoVaS methods more formally; see \cite{diebold2002comparing} for the technical details of the DM-test\footnote{We perform the DM-test with the function \textit{DM-test} in the \textit{R} package \textit{multDM}.}. The DM-test results on comparing GARCHX and GAX-NoVaS for forecasting three simulated datasets are tabularized in \cref{Table:DM_test_short}. These tests further verify the advantage of our methods on forecasting with exogenous variables, especially for a long-prediction horizon. 

\begin{table}[htbp]
  \centering
  \caption{DM-test results on simulated datasets.}
    \resizebox{\textwidth}{!}{
    \begin{tabular}{lccccccccc}
    \toprule
          & \multicolumn{3}{c}{\textbf{Model-1}} & \multicolumn{3}{c}{\textbf{Model-2}} & \multicolumn{3}{c}{\textbf{Model-3}} \\[2pt]
    \hline
    \textbf{Prediction steps} & 1     & 5     & 20    & 1     & 5     & 20 & 1     & 5     & 20 \\
    \textbf{GARCHX} &       &       &       &       &       &       &       &       &  \\
    \textbf{GAX-NoVaS} & 0.009 & 0.020 & 0.000 & 0.038 & 0.107 & 0.000 & 0.172 & 1.000 & 0.102  \\
    \bottomrule
    \end{tabular}}%
\raggedright

\noindent{\footnotesize{Note: The values in the different rows are one-sided $p$-values of the DM-test on the prediction error of GAX-NoVaS and GARCHX. For the DM-test here, the alternative hypothesis is that the former method GAX-NoVaS is more accurate than the latter one, GARCHX. }} 

  \label{Table:DM_test_short}%
\end{table}%
\FloatBarrier

\section{Empirical analyses with real data}\label{Sec:empiricalstudies}
A summarizing note of our findings in \cref{Sec:Simulations} reads that the GAX-NoVaS method performs better than standard GARCH-type methods, especially for long-term time aggregated predictions. In this section, we deploy an interesting data analysis to exhibit how our method could shed light on the role of geopolitical risks in
forecasting volatility with real-world data. We start by describing the data below. 
\subsection{Data description}\label{Subsec:datadescription}
The ongoing Ukraine-Russia has led many countries in Europe, particularly Germany, to adjust their military and security, as well as energy policies in light of new geopolitical risks, with such adjustments entailing large costs to the macroeconomy and financial markets, as depicted by \cite{grebe2024}. In this regard, these authors, first, assemble a data set of more than eight million German Twitter posts related to the war in Ukraine to construct a daily index of uncertainty about the war as perceived by German Twitter based on using state-of-the-art methods of textual analysis. \cite{grebe2024} show that an increase in uncertainty has strong effects on financial markets, associated with a significant decline in economic activity as well as an increase in expected inflation. We utilize this index (Ukraine)\footnote{The data is available for download from: \url{https://www.uni-giessen.de/de/fbz/fb02/fb/professuren/vwl/tillmann/forschung/ukraine-uncertainty-index}.} in our empirical analysis to forecast stock market volatility of not only Germany, but two of its neighbors namely, Austria and Switzerland, over the daily period of 1st January, 2021 to 28th February, 2023. The national stock market indexes (ATX (Austria), DAX (Germany), SMI (Switzerland)) of these three countries, for which we compute log-returns to feed into our volatility models were derived from the Bloomberg terminal. With the focus being on geopolitical risks, we also utilized the daily newspapers-based geopolitical risks index (GPRD) of \cite{caldara2022}\footnote{The data can be accessed from: \url{https://www.matteoiacoviello.com/gpr.htm}.}, which, in turn, allowed us to analyze a longer data sample covering 2nd January, 2006 to 10th August, 2023. The starting date of this longer sample, and the choice of these three countries, were also motivated by the availability of Google searches-based daily data on economic activity (Trend) for all three countries, and inflation (Inflation) for Germany and Switzerland,\footnote{The data can be downloaded from: \url{https://www.trendecon.org/}.} which are used as additional predictors to ensure that our results are not only limited to geopolitical risks.

We present three log-return series of Germany, Switzerland and Austria from 2nd January, 2006 to 10th August, 2023 in \cref{fig:Germanylong,fig:Switzerlandlong,fig:Austrialong}. The volatility clustering phenomena observed in all plots reveals the heteroskedasticity within these three series. To investigate the property of three long return series, we provide the summary statistics, e.g., the mean, skewness, and kurtosis in \cref{Table:summarystatistics}. To verify the heteroskedasticity with all series more directly, we split the whole time period into four equal-length sub-periods and denote the sample variance of all four sub-periods by $V_i$, $i = 1,\ldots, 4$. These statistics are also provided in \cref{Table:summarystatistics}. Towards statistical tests, we also perform modified Ljung-Box (m-LB) and ARCH Lagrange Multiplier (ALM) tests to check the autocorrelation and ARCH effects of squared return series. For the m-LB test, we consider the lag order 20. For the ALM test, we consider the maximum lag order 10. These two tests are performed in \textit{R} with functions \textit{lbtest} and \textit{Lm.test}, respectively. The p-values of tests are presented. Summarizing \cref{Table:summarystatistics}, the large kurtosis values indicate the heavy-tailed property for all three log-return series. The variance of return series in different time regions changes notably, indicating the heteroskedasticity. The ALM test with a pretty small p-value also confirms the heteroskedasticity for all return series. The m-LB test shows strong evidence of autocorrelation within all squared return series.
\begin{table}[htbp]
  \centering
  \caption{Summary statistics of three long return series.}
    \begin{tabular}{lccccccccc}
    \toprule
     Returns series  & Mean &  Skew. &  Kurt. & $V_1$ & $V_2$ & $V_3$ & $V_4$  & m-LB    &  ALM  \\[2pt]
    \hline
    \textbf{Germany} &  0.01 & -0.24 &  11.37& 2.51 & 1.52 &1.22  & 1.79 & 0.00 & 0.00    \\
    \textbf{Switzerland} & 0.01  & -0.43 & 12.59 & 1.83 & 0.83 & 0.88 & 0.95 & 0.00 & 0.00   \\
    \textbf{Austria} & -0.01  & -0.41 & 10.26 &  4.78 & 2.15 & 1.81 &  3.27 &  0.00& 0.00 \\
    \bottomrule
    \end{tabular}%
\raggedright

\noindent{\footnotesize{Note: columns $V_i$, $i = 1,\ldots, 4$ represent the sample variance of each long series on four equal-length sub-periods splitted from the whole period 01/02/2006 to 08/10/2023. The column ALM represents the p-value of the ALM test with the maximum lag order being 10; the column m-LB represents the p-value of the m-LB test at the lag order being 20; 0.00 indicates the p-value is less than $3\times 10^{-16}$ for these two tests. Skew. and Kurt. represent the skewness and kurtosis respectively. }} 

  \label{Table:summarystatistics}%
\end{table}%

To show the fluctuations behind the Ukraine index and GPRD index, we present two plots in \cref{fig:Ukraineplot,fig:GPRDplot} below. As one can see from there, these two indices fluctuate severely around the beginning of 2022 which corresponds with the real-world event. Later, we attempt to use this information to forecast the stock market volatility of three countries.

\begin{figure}
    \centering
    \includegraphics[scale = 0.65]{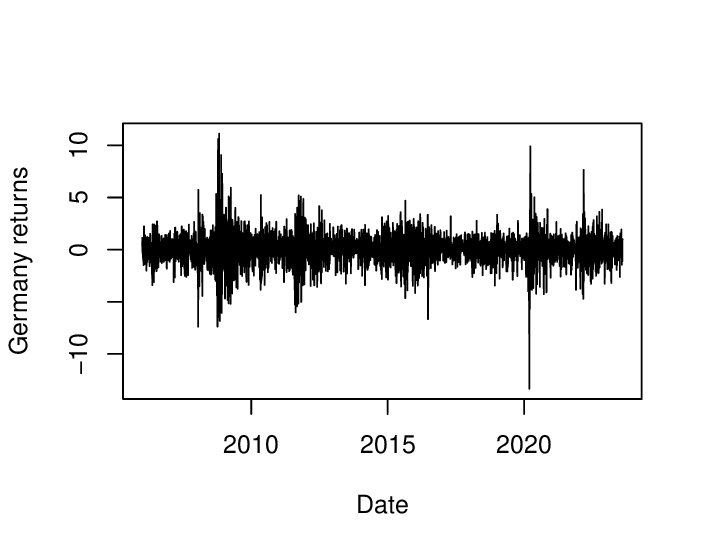}
    \caption{The log-returns of Germany's national stock market index from 01/02/2006 to 08/10/2023.}
    \label{fig:Germanylong}
\end{figure}

\begin{figure}
    \centering
    \includegraphics[scale = 0.65]{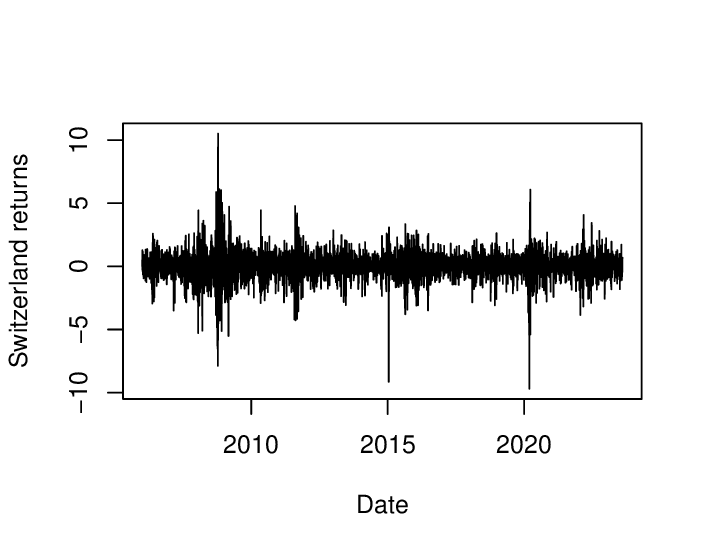}
    \caption{The log-returns of Switzerland's national stock market index from 01/02/2006 to 08/10/2023.}
    \label{fig:Switzerlandlong}
\end{figure}

\begin{figure}
    \centering
    \includegraphics[scale = 0.65]{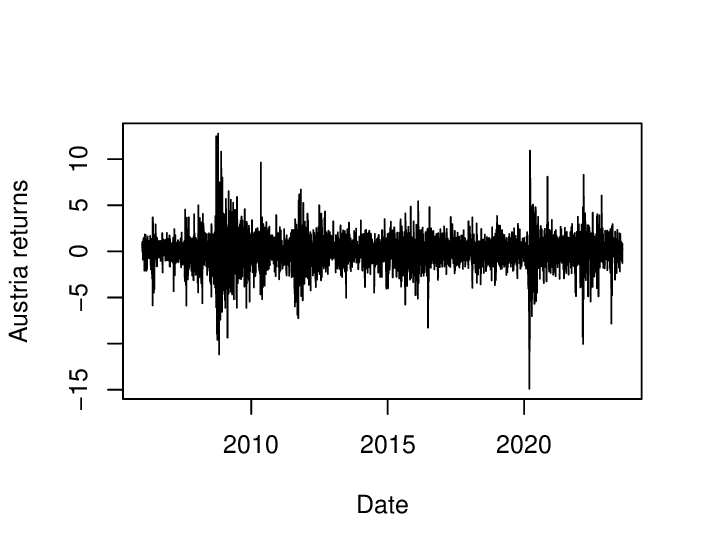}
    \caption{The log-returns of Austria's national stock market index from 01/02/2006 to 08/10/2023.}
    \label{fig:Austrialong}
\end{figure}

\begin{figure}
    \centering
    \includegraphics[scale = 0.65]{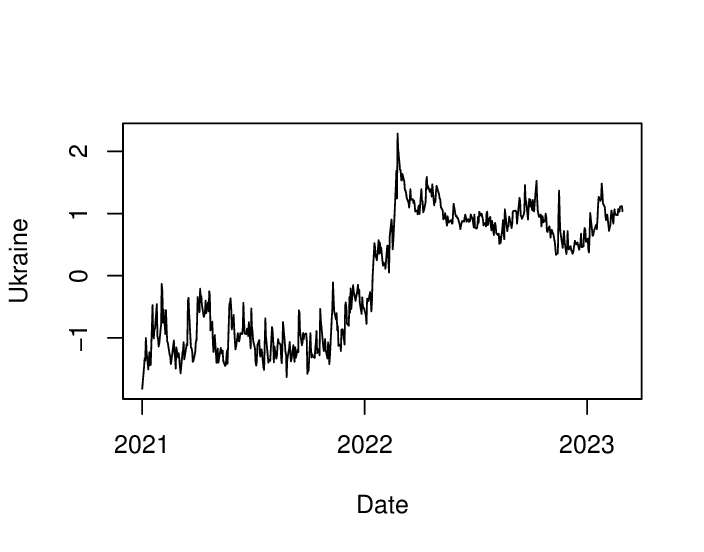}
    \caption{The exogenous variable Ukraine index from 01/01/2021 to 28/02/2023. This is corresponding to the time period of our short data .}
    \label{fig:Ukraineplot}
\end{figure}

\begin{figure}
    \centering
    \includegraphics[scale = 0.65]{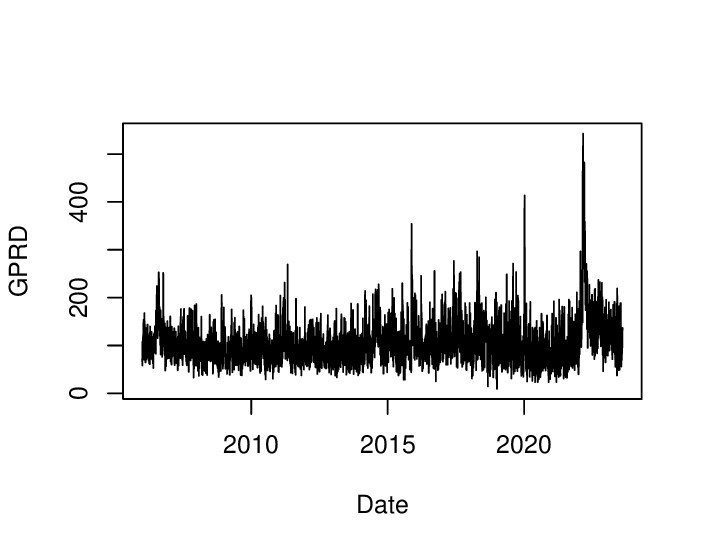}
    \caption{The exogenous variable GPRD index from 01/02/2006 to 08/10/2023.}
    \label{fig:GPRDplot}
\end{figure}

\FloatBarrier
\subsection{Empirical results}
We first consider the forecasting exercise with the short period (1st January,
2021 to 28th February, 2023) data described in \cref{Subsec:datadescription}. Then, the analysis of three long returns series is given in \cref{Subsubsec:long}. 
\subsubsection{Short data}
To compare the performance of GARCH, GARCHX and GAX-NoVaS methods, we still apply the time aggregated prediction metric described in \cref{Subsec:timeaggregatedmetric}. We consider $h = 1,5, 20$ and use a 250-size moving window, which is about 1 year of daily data. We start the empirical analysis with the short period of data and we take the Ukraine index as the exogenous predictor. The MSPE results are summarized in \cref{Table:empiricalre_shortMSPE}.

\begin{table}[htbp]
  \centering
  \caption{MSPE ratios of different methods on three short datasets.}
    \resizebox{\textwidth}{!}{
    \begin{tabular}{lccccccccc}
    \toprule
          & \multicolumn{3}{c}{\textbf{S-Germany}} & \multicolumn{3}{c}{\textbf{S-Austria}} & \multicolumn{3}{c}{\textbf{S-Switzerland}} \\[2pt]
    \hline
       \textbf{Prediction steps} & 1     & 5     & 20    & 1     & 5     & 20 & 1     & 5     & 20 \\
    \textbf{GA} & 1.000 & 1.000 & 1.000 & 1.000 & 1.000 & 1.000 & 1.000 & 1.000 & 1.000\\
    \textbf{GAX-Ukraine} & 1.011 & 1.023 & 0.989 & 1.002 & 1.003 & 0.994 & 1.003 & 1.008 & 1.005\\
    \textbf{GAX-NoVaS-Ukraine} & 1.038 & 1.141 & 0.898 & 1.017 & 0.984 & 0.822 & 1.010 & 1.100 & 0.879 \\
    \bottomrule
    \end{tabular}}%
    \raggedright

\noindent{\footnotesize{Note: ``S'' represents ``short''. To simplify the presentation, we compute the ratio of MSPE of different methods, i.e., we divide all MSPE for each prediction horizon and model by the MSPE of the GARCH (benchmark) method.}} 
  \label{Table:empiricalre_shortMSPE}%
\end{table}%
From \cref{Table:empiricalre_shortMSPE}, we can see that the GAX-NoVaS method can bring some large improvements, especially for long-horizon aggregated predictions. Meanwhile, it seems that the GARCHX and GARCH models have indistinguishable performance. However, the GAX-NoVaS method is generally better than both GARCH-type methods. The DM-test results on comparing GARCHX and GAX-NoVaS for forecasting short real-world data are tabularized in \cref{Table:DM_test_short}, which reveals the significant advantage of GAX-NoVaS for long-horizon predictions, especially for the 20-step-ahead predictions of Short Austria data.

\begin{table}[htbp]
  \centering
  \caption{DM-test results on predictions of short datasets.}
    \resizebox{\textwidth}{!}{
    \begin{tabular}{lccccccccc}
    \toprule
          & \multicolumn{3}{c}{\textbf{S-Germany}} & \multicolumn{3}{c}{\textbf{S-Austria}} & \multicolumn{3}{c}{\textbf{S-Switzerland}} \\[2pt]
    \hline
    \textbf{Prediction steps} & 1     & 5     & 20    & 1     & 5     & 20 & 1     & 5     & 20 \\
    \textbf{GARCHX-Ukraine} &       &       &       &       &       &       &       &       &  \\
    \textbf{GAX-NoVaS-Ukraine} & 0.616 & 0.727 & 0.306 & 0.578 & 0.410 & 0.030 & 0.549 & 0.860 & 0.118 \\
    \bottomrule
    \end{tabular}}%
\raggedright

\noindent{\footnotesize{Note: The values in the different rows are one-sided $p$-values of the DM-test on the prediction error of GAX-NoVaS and GARCHX. For the DM-test here, the alternative hypothesis is that the former method GAX-NoVaS is more accurate than the latter one GARCHX. }} 

  \label{Table:DM_test_short}%
\end{table}%

\subsubsection{Long data}\label{Subsubsec:long}
We continue our real data analysis with long datasets (2nd January, 2006 to 10th August, 2023). We also apply more exogenous variables. Similar to the analysis procedure for short data, we present MSPE ratios and corresponding DM-test results of GARCHX and GAX-NoVaS in \cref{Table:empiricalre_longMSPE,Table:DM_test_long}. Generally speaking, the GAX-NoVaS method still dominates the other two GARCH-type methods, and this superiority is verified to be significant by the DM-test. 
\begin{table}[htbp]
  \centering
  \caption{MSPE ratios of different methods on three long datasets.}
    \resizebox{\textwidth}{!}{
    \begin{tabular}{lccccccccc}
    \toprule
          & \multicolumn{3}{c}{\textbf{L-Germany}} & \multicolumn{3}{c}{\textbf{L-Austria}} & \multicolumn{3}{c}{\textbf{L-Switzerland}} \\[2pt]
    \hline
    \textbf{Prediction steps} & 1     & 5     & 20    & 1     & 5     & 20    & 1     & 5     & 20 \\
    \textbf{GA} & 1.000 & 1.000 & 1.000 & 1.000 & 1.000 & 1.000 & 1.000 & 1.000 & 1.000 \\
    \textbf{GA-NoVaS} & 0.990 & 0.922 & 0.518 & 1.062 & 1.041 & 0.715 & 1.010 & 0.786 & 0.128 \\
    \textbf{GAX-GPRD} & 1.010 & 1.024 & 0.998 & 0.999 & 0.990 & 0.964 & 1.003 & 0.994 & 0.797 \\
    \textbf{GAX-NoVaS-GPRD} & 1.000 & 0.953 & 0.557 & 1.076 & 1.121 & 0.831 & 1.005 & 0.799 & 0.133 \\
    \textbf{GAX-Inflation} & 1.003 & 1.016 & 0.974 &  &  &  & 0.996 & 0.965 & 0.858 \\
    \textbf{GAX-NoVaS-Inflation} & 1.001 & 0.949 & 0.551 &  &  &  & 1.010 & 0.801 & 0.131\\
    \textbf{GAX-Trend} & 0.987 & 0.979 & 1.062 & 0.982 & 0.912 & 0.654 & 0.996 & 0.937 & 0.594 \\
    \textbf{GAX-NoVaS-Trend} & 0.998 & 0.949 & 0.540 & 1.077 & 1.169 & 0.866 & 1.019 & 0.828 & 0.138  \\
   \bottomrule
    \end{tabular}}%
    \raggedright

\noindent{\footnotesize{Note: ``L'' represents ``long''. To simplify the presentation, we compute the ratio of MSPE of different methods, i.e., we divide all MSPE for each prediction horizon and model by the MSPE of the GARCH (benchmark) method.}} 
  \label{Table:empiricalre_longMSPE}%
\end{table}%

\begin{table}[htbp]
  \centering
  \caption{DM-test results on predictions of long datasets.}
    \resizebox{\textwidth}{!}{
    \begin{tabular}{lccccccccc}
    \toprule
          & \multicolumn{3}{c}{\textbf{L-Germany}} & \multicolumn{3}{c}{\textbf{L-Austria}} & \multicolumn{3}{c}{\textbf{L-Switzerland}} \\
                  \hline
    \textbf{Prediction steps} & 1     & 5     & 20    & 1     & 5     & 20    & 1     & 5     & 20 \\
    \textbf{GARCHX-GPRD} &       &       &       &       &       &       &       &       &  \\
    \textbf{GAX-NoVaS-GPRD} & 0.427 & 0.306 & 0.012 & 0.950 & 0.957 & 0.068 & 0.516 & 0.055 & 0.008 \\
    \textbf{GARCHX-Inflation} &       &       &       &       &       &       &       &       &  \\
    \textbf{GAX-NoVaS-Inflation} & 0.484 & 0.298 & 0.007 &       &       &       & 0.595 & 0.082 & 0.029 \\
    \textbf{GARCHX-Trend} &       &       &       &       &       &       &       &       &  \\
    \textbf{GAX-NoVaS-Trend} & 0.606 & 0.420 & 0.017 & 0.995 & 1.000 & 1.000 & 0.659 & 0.122 & 0.027 \\
    \bottomrule
    \end{tabular}}%
\raggedright

\noindent{\footnotesize{Note: The values in the different rows are one-sided $p$-values of the DM-test on the prediction error of GAX-NoVaS and GARCHX. For the DM-test here, the alternative hypothesis is that the former method GAX-NoVaS is more accurate than the latter one GARCHX. }} 

  \label{Table:DM_test_long}%
\end{table}%

\FloatBarrier
\section{Conclusion}\label{Sec:Conclusion}

We extend the current NoVaS prediction method to the realm of prediction with exogenous variables. We provide the theoretical foundation to guarantee the feasibility of applying Model-free prediction. Inspired by the GARCHX model, we propose a specifically designed model-free/model-based method namely GAX-NoVaS prediction. The dominance of GAX-NoVaS on the classical GAX-NoVaS method is verified by simulation and empirical datasets. Also, such an advantage is not only exhibited by some MSPE metric but we also show some statistical significance through the parlance of classical DM tests. 

We should also mention that going far beyond GAX-NoVaS method might have limitations if the model becomes increasingly complex. Recall that the great performance of the GAX-NoVaS method relies on a successful transformation, the satisfied transformation may not be achievable if the underlying time series is very complicated. However, there is a growing literature on forecasting using more non-parametric neural network based models. It will be an interesting future work that combines the idea of model-free prediction with the state-of-the-art machine learning method, such as Deep neural network (DNN), convolutional neural network (CNN) or LSTM etc. Finally, in the field of binary/categorical/count data INGARCH models have recently garnered significant attention both from theoretical and applied researchers. One could potentially also think of a Model-free INGARCH-X type model and try to integrate the idea of exogenous covariates into it and build a new forecasting framework to challenge the existing ones. In short, our paper remains the first paper to propose this idea of model-free transformation-based forecasting focused on GARCHX type models but the scope of extending this to several directions is ample.

\clearpage

\clearpage
\bibliographystyle{apalike}
\bibliography{refs}
\end{document}